\documentclass[11pt, letterpaper]{article}
\usepackage[scale=0.75]{geometry}

\usepackage{amsmath}
\usepackage{amsthm}
\usepackage{amssymb}
\usepackage{enumerate}
\usepackage{mathtools}

\makeatletter
\newcommand{\newreptheorem}[2]{%
\newtheorem*{rep@#1}{\rep@title}%
\newenvironment{#1}[1]{%
 \def\rep@title{#2 \ref{##1}}%
 \begin{rep@#1}}%
 {\end{rep@#1}}}
\makeatother

\newtheoremstyle{plainrep}
  {\topsep}   
  {\topsep}   
  {\itshape}  
  {0pt}       
  {\itshape} 
  {.}         
  {5pt plus 1pt minus 1pt} 
  {}          

\theoremstyle{plain}
\newtheorem{theorem}{Theorem}[section]
\newtheorem{lemma}[theorem]{Lemma}
\newtheorem{claim}[theorem]{Claim}
\theoremstyle{definition}
\newtheorem{definition}{Definition}
\theoremstyle{remark}
\newtheorem{remark}{Remark}
\theoremstyle{plainrep}
\newreptheorem{reptheorem}{Theorem}

\DeclarePairedDelimiter{\bracket}{[}{]}
\DeclarePairedDelimiter{\abs}{\lvert}{\rvert}
\DeclarePairedDelimiter{\norm}{\lVert}{\rVert}
\DeclarePairedDelimiter{\floor}{\lfloor}{\rfloor}
\DeclarePairedDelimiter{\set}{\{}{\}}

\newcommand{\Eop}{\mathop{\mathbf{E}}\nolimits}
\makeatletter
\newcommand{\E}{\@ifstar\E@star\E@withoutstar}
\newcommand{\E@withoutstar}{\@ifnextchar[\E@size\E@plain}
\newcommand{\E@star}{\@ifnextchar_\E@starsub\E@starnosub}
\def\E@size[#1]{\@ifnextchar_\E@sizesub{#1}\E@sizenosub{#1}}
\newcommand{\E@plain}{\@ifnextchar_\E@plainsub\E@plainnosub}
\def\E@starsub_#1{\Eop_{#1}\bracket*}
\newcommand{\E@starnosub}{\Eop\bracket*}
\def\E@sizesub#1_#2{\Eop_{#2}\bracket[#1]}
\newcommand{\E@sizenosub}[1]{\Eop\bracket[#1]}
\def\E@plainsub_#1{\Eop_{#1}\bracket}
\newcommand{\E@plainnosub}{\Eop\bracket}
\makeatother

\date{}
\newcommand{\NN}{\mathbb{N}}

\newcommand{\A}{{\mathcal{A}}}
\newcommand{\GP}{{\mathcal{GP}}}
\newcommand{\GPk}{{\mathcal{GP}_k}}
\newcommand{\Pk}{{\mathcal{P}}}
\newcommand{\U}{{\mathcal{U}}}
\newcommand{\Uk}{{\mathcal{U}_k}}
\newcommand{\deq}{\stackrel{\textrm{def}}{=}}

\newcommand{\zo}{\{0,1\}}
\newcommand{\zon}{\{0,1\}^n}

\newcommand{\zonk}{\{0,1\}^{n\times k}}

\newcommand{\zolk}{\{0,1\}^{l\times k}}

\newcommand{\Vkt}{\mathcal{V}_{k,t}}
\newcommand{\tlr}{\tilde{r}}
\newcommand{\tlR}{\tilde{R}}

\newcommand{\vR}{\vec{R}}
\newcommand{\vr}{\vec{r}}

\newcommand{\GF}{\mathrm{GF}}

\title{Shared Randomness and Quantum Communication \\
  in the Multi-Party Model}

\author{%
  Dmitry Gavinsky\thanks{NEC Laboratories America, Inc., Princeton, NJ, USA.} \qquad
  Tsuyoshi Ito\footnotemark[1] \qquad
  Guoming Wang\footnotemark[1] \thanks{Computer Science Division, University of California, Berkeley, Berkeley, CA, USA.}
}

\begin{document}

\maketitle

\vspace*{-4ex}

\begin{abstract}
  We study shared randomness in the context of multi-party number-in-hand communication protocols
  in the simultaneous message passing model.
  We show that with three or more players,
  shared randomness exhibits new interesting properties
  that have no direct analogues in the two-party case.

  First, we demonstrate a hierarchy of modes of shared randomness,
  with the usual shared randomness where all parties access the same random string
  as the strongest form in the hierarchy.
  We show exponential separations between its levels,
  and some of our bounds may be of independent interest.
  For example, we show that the equality function can be solved
  by a protocol of constant length using the weakest form of shared randomness,
  which we call \emph{XOR-shared randomness}.

  Second, we show that quantum communication cannot replace shared randomness in the $k$-party case,
  where~$k\ge3$ is any constant.
  We demonstrate a promise function~$\GPk$
  that can be computed by a classical protocol of constant length
  when (the strongest form of) shared randomness is available,
  but any quantum protocol without shared randomness
  must send~$n^{\Omega(1)}$ qubits to compute it.
  Moreover, the quantum complexity of~$\GPk$ remains~$n^{\Omega(1)}$
  even if the ``second strongest'' mode of shared randomness is available.
  While a somewhat similar separation was already known in the two-party case,
  in the multi-party case our statement is qualitatively stronger:
  \begin{itemize}
  \item
    In the two-party case, only a relational communication problem
    with similar properties is known.
  \item
    In the two-party case, the gap between the two complexities of a problem can be at most exponential,
    as it is known that~$2^{O(c)}\log n$ qubits can always replace shared randomness in any~$c$-bit protocol.
    Our bounds imply that with quantum communication alone, in general,
    it is not possible to simulate efficiently even a three-bit three-party classical protocol
    that uses shared randomness.
  \end{itemize}
\end{abstract}

\section{Introduction}

The area of communication complexity deals with the amount of communication
required for solving computational problems with distributed input.
In the two-party \emph{simultaneous message passing (SMP)} setting of communication complexity,
two \emph{players} Alice and Bob receive inputs~$x$ and~$y$, respectively,
and each sends a message to a third party, the \emph{referee}.
Using those messages, the referee computes the output value.
When the goal is to compute certain function~$f$,
the success is measured by the probability
that the output value of a \emph{communication protocol} equals $f(x,y)$.
The \emph{cost} of a communication protocol is the total number of bits
sent by the players to the referee.

Shared randomness is a crucial resource in communication complexity.
When Alice and Bob have it, they can use a mixed strategy in order to compute $f$;
in particular, the minimax principle implies
that the worst-case and the average-case complexities are equal in this case.
It is known that without shared randomness the model becomes considerably weaker,
and the gap between the worst-case and the average-case complexities
of a communication problem can be arbitrary large
(constant vs.~$\Omega(\sqrt{n})$ in the case of the equality function,
as shown by Newman and Szegedy~\cite{NS96_Pub}).

The SMP model of communication is the weakest among those that have been studied widely.
Nevertheless, it is arguably the right model to look at when the goal is to investigate shared randomness.
That is because whenever communication between the players is possible
(which is the case for all other commonly studied models, but not for SMP),
the first player can append~$O(\log n)$ bits of \emph{private} randomness
to the first message that is sent to the others,
and that would not affect the cost of the protocol significantly,
as poly-logarithmic cost is usually viewed as efficient.
Those random bits are now known to all the participants,
and can be used in place of shared randomness.
It is known due to Newman~\cite{N91_Pri} that~$O(\log n)$ bits of shared randomness are always enough;
therefore, shared randomness does not make much difference
in any model that allows direct communication between the players.

In this paper we study shared randomness in the context of the \emph{multi-party version} of the SMP model,
where the number of players~$k$ is three or larger (the referee is not counted as a player),
the input has~$k$ fragments and each fragment is known to exactly one player---%
this regime of distributing input between the players is usually called ``number in hand.''
This model can be viewed as a natural generalization of the two-party model.

We demonstrate several interesting (and somewhat surprising) properties of shared randomness
when the number of players is at least three, that have no direct analogues in the two-party case.

\subsection{Previous work}

In~\cite{Y03_On_th}, Yao generalized the technique of \emph{quantum fingerprints}~\cite{BCWW01_Qu}
to show that every classical two-party SMP protocol that uses shared randomness and sends~$c$ bits
can be simulated by a quantum protocol without shared randomness that sends~$2^{O(c)}\log n$ qubits.
This naturally raised the question whether quantum communication can always replace shared randomness---%
that is, whether any communication problem that can be solved
by a classical SMP protocol of poly-logarithmic length using shared randomness
can also be solved by a quantum protocol of poly-logarithmic length without shared randomness.

The question was addressed by Gavinsky, Kempe, Regev and de~Wolf in~\cite{GKRW06_Bou},
where they demonstrated a two-party \emph{relational} communication problem
that can be solved by a classical protocol of cost~$O(\log n)$ that uses shared randomness,
but requires~$n^{\Omega(1)}$ qubits in order to be solved by a quantum protocol without shared randomness.
In the same work a question was posed whether a similar separation is possible via a \emph{functional} problem.

\section{Our results}

Two main results of this work are the following.

First, we establish a hierarchy of modes of shared randomness (Section~\ref{s_srand}).
In the $k$-party SMP model, we consider \emph{$t$-shared randomness} for~$2\le t\le k$,
where every set of~$t$ players shares a random string.
The~$k$-shared randomness is the usual, unrestricted shared randomness
and the strongest mode in the hierarchy,
and a smaller value of~$t$ gives a weaker form of shared randomness.
The~$(k-1)$-shared randomness could be also called ``randomness on the forehead.''
Below~$2$-shared randomness, we also consider an even more restricted mode of shared randomness
which we call \emph{XOR-shared randomness},
where the~$k$ players receive uniformly random~$k$-tuples of bits whose parity is~$0$.
The precise definitions of these modes of shared randomness
will be given in Section~\ref{s_srand}.
We will show that this is a proper hierarchy; i.e., we show exponential separations between its levels.

One of the problems that we study in this context is the multi-party equality function,
and we show (Claim~\ref{c_k-eq}) that it can be solved by a protocol of constant length
that uses XOR-shared randomness.
We believe that this result might be of independent interest, due to the importance of the equality function.

Second, we demonstrate a promise function
whose classical communication complexity is constant if the strongest form of shared randomness is available,
but whose quantum communication complexity is~$n^{\Omega(1)}$ if no shared randomness is available (Section~\ref{s_vs}).
Moreover, the quantum complexity remains~$n^{\Omega(1)}$
even if the protocol can use $(k-1)$-shared randomness
(randomness on the forehead).

Our second result is closely related to~\cite{GKRW06_Bou}:
We demonstrate a \emph{promise function}
that can be solved efficiently in the classical model with shared randomness,
but not in the quantum model without it.
This answers the main open problem posed in~\cite{GKRW06_Bou} for the case of three or more players.
We note that the question remains wide open in the two-player case.

Our second result is also related to the aforementioned work by Yao~\cite{Y03_On_th},
where it was shown, informally speaking, that shared randomness can be replaced by quantum communication
with (at most) exponential overhead.
In this work we demonstrate a (functional) communication problem
that can be solved by a three-bit three-party classical protocol with shared randomness
but requires~$n^{\Omega(1)}$ qubits without shared randomness
(or even with randomness on the forehead).
Accordingly, the possibility to simulate shared randomness by quantum communication
is a unique feature of the two-party model;
with more than two players, the possible advantage of shared randomness over quantum communication
is not bounded by any function.

\subsection{Technical statements}

In the first part, we prove the following.

\newcommand{\theohie}{%
  Let $k\ge3$.
  Then,
  \begin{itemize}
  \item
    For each $t\in\{3,\dots,k\}$, there exists a~$k$-party promise function
    that can be solved by a protocol of cost~$t$ in the SMP model
    with classical communication and~$t$-shared randomness
    but requires~$\Omega(tn^{1/t})$ qubits in the SMP model
    with quantum communication and $(t-1)$-shared randomness.
  \item
    There exists a~$k$-party total function
    that can be solved by a protocol of constant cost
    in the classical SMP model with $2$-shared randomness
    but requires~$\Omega(\sqrt{n})$ bits of communication
    in the classical SMP model with XOR-shared randomness.
  \item
    There exists a~$k$-party total function
    that can be solved by a protocol of constant cost
    in the classical SMP model with XOR-shared randomness
    but requires~$\Omega(\sqrt{n})$ bits of communication
    in the classical SMP model without shared randomness.
  \end{itemize}
}

\begin{theorem}
  \label{t_hie}
  \theohie
\end{theorem}

In the second part of this work,
we study the following natural communication problem.
For a bit string~$x$, we denote by~$\abs{x}$ its Hamming weight,
i.e.\ the number of 1s in~$x$.

\begin{definition}[Gap-Parity]
  Let~$x_1,\dots,x_k$ be~$n$-bit strings
  such that~$\abs{x_1\oplus\dots\oplus x_k}\notin[n/3,2n/3]$.
  Then we define~$\GPk(x_1,\dots,x_k)=0$
  if~$\abs{x_1\oplus\dots\oplus x_k}<n/2$
  and~$\GPk(x_1,\dots,x_k)=1$ otherwise.
\end{definition}

Note that in the SMP model with classical communication and shared randomness,
$\GPk$ has a trivial solution, where each player sends only one bit to the referee.

We will demonstrate that for~$k\ge3$,
$\GPk$ cannot be solved efficiently by a quantum protocol without shared randomness.
Moreover, we show that the Gap-Parity problem has no efficient solution with quantum communication
even with randomness on the forehead
(cf.\ Section~\ref{s_srand}).

\newcommand{\theomain}{%
  Let~$k\ge3$.
  Using shared randomness, the~$k$-party promise function~$\GPk$ can be solved
  by a classical SMP protocol of cost~$k$ where each player sends a bit.
  For~$2\le t\le k-1$,
  in the SMP model with quantum communication with~$t$-shared randomness,
  the complexity of~$\GPk$ is~$\Omega(kn^{1-t/k})$.
  In particular, in the SMP model with quantum communication without shared randomness,
  the complexity of~$\GPk$ is~$\Omega(kn^{1-2/k})$.
}

\begin{theorem}
  \label{t_main}
  \theomain
\end{theorem}

\section{Preliminaries}
  \label{s_prel}

For any~$n$-dimensional vector~$v$,
we will write~$v(j)$ to denote its~$j$th coordinate,
and for any~$S\subseteq[n]$ we will use~$v_S$
to denote the restriction of~$v$ to the coordinates that are elements of~$S$.
We use~$\bar{0}$ or~$\bar{1}$ to denote the vectors of all $0$s or all $1$s, respectively,
when its length is clear from the context.

For any finite set~$W$, let~$\U(W)$ denote the uniform distribution over the elements of~$W$.

\subsection{Quantum measurements}

Unless stated otherwise, we will represent quantum states by their density matrices.
We will write~$\E_i{\sigma_i}$ or even $\E{\sigma_i}$ to denote the mixed state~$(1/k)\sum_i\sigma_i$.

Given a matrix~$M$, we denote by~$\norm{M}_1$ the \emph{trace norm} of~$M$,
defined as the sum of the singular values of~$M$.
It is known that given two quantum states~$\sigma_0$ and~$\sigma_1$,
the optimal probability with which a quantum measurement
can correctly distinguish between~$\sigma_0$ and~$\sigma_1$
equals $1/2+\norm{\sigma_0-\sigma_1}_1/4$.

We will use the following special case of the ``random access code argument''~\cite[Lemma~2.2]{GKRW06_Bou},
which is a slight generalization of~\cite[Theorem~2.3]{N99_Op_Lo} (see also~\cite{ANTV02_De}).

\begin{claim}
  \label{c_rac}
  Let~$X \sim \U(\zon)$.
  Suppose for each instantiation~$X=x$ we have a quantum state~$\rho_x$ of~$q$ qubits.
  Let~$\sigma_a^j$ be the expectation of~$\rho_X$ conditional upon~$X(j)=a$,
  for~$j\in[n]$ and~$a \in \zo$.
  Then~$\sum_{j=1}^n\norm{\sigma_0^j-\sigma_1^j}_1^2 \in O(q)$.
\end{claim}

\begin{proof}
  Let~$h(p)$ be the binary entropy function: $h(p)=-p\log_2 p-(1-p)\log_2(1-p)$.
  Lemma~2.2 of~\cite{GKRW06_Bou} implies that under the assumption of the claim,
  it holds that~%
  $\sum_{j=1}^n \bigl(1-h(1/2-\norm{\sigma_0^j-\sigma_1^j}_1/4)\bigr) \le q$.
  The claim follows because~$1-h(1/2-x/4)\ge x^2/(8\ln2)$ for~$0\le x\le2$.
\end{proof}

Now let us consider the situation where a quantum measurement is performed
in order to predict the parity of several independent binary variables.
\begin{claim}
  \label{c_xor}
  For every~$i\in[m]$, let~$\sigma_0^i$ and~$\sigma_1^i$ be quantum states of equal dimension.
  For~$a\in\zo$, let~%
  $\rho_a\deq\E_{\alpha_1\oplus\dots\oplus\alpha_m=a}{\sigma_{\alpha_1}^1\otimes\dots\otimes\sigma_{\alpha_m}^m}$.
  Then~%
  $\norm{\rho_0-\rho_1}_1
   =(1/2^{m-1})\prod_{i=1}^m\norm{\sigma_0^i-\sigma_1^i}_1$.
\end{claim}

\begin{proof}
  Write:
  \[
    \rho_0-\rho_1=
    \frac{1}{2^{m-1}}(\sigma_0^1-\sigma_1^1)\otimes\dots\otimes
    (\sigma_0^m-\sigma_1^m),
  \]
  and the claim follows from the fact that the trace norm is multiplicative with respect to the tensor product.
\end{proof}

\subsection{Communication complexity}

In this work we are interested in the following model of communication complexity.

\begin{definition}[Multi-party SMP]
  The~$k$-party simultaneous message passing (SMP) model
  involves~$k+1$ parties: $k$ players~$\A_1,\dots,\A_k$ and a referee.
  For every~$i\in[k]$, player~$\A_i$ gets input~$x_i$.
  They each send one message to the referee,
  who uses the content of all~$k$ messages to compute the output value.

  A communication protocol describes the action of each participant.
  The \emph{cost} or \emph{complexity} of a protocol
  is the total length of the messages sent by players~$A_1,\dots,A_k$ to the referee.
  We say that a protocol solves a computational problem defined over~$k$ input values
  if the referee gives a correct answer with probability at least~$2/3$ for each possible input.
\end{definition}

In this paper we will consider several further modifications of the SMP model:
\begin{itemize}
\item
  In the \emph{quantum} SMP model,
  the players~$\A_1,\dots,\A_k$ are allowed to send quantum messages,
  and the referee can perform any quantum measurement in order to determine the output value.
\item
  In the SMP model \emph{with shared randomness},
  the players~$\A_1,\dots,\A_k$ have free access to the same string of random bits
  that were chosen independently from the input values.
\item
  In Section~\ref{s_srand},
  we will define a \emph{hierarchy of modes of shared randomness}
  in multi-party protocols (where the strongest mode is the standard one, as described above).
  We demonstrate exponential separations between the levels of the hierarchy
  (i.e., the hierarchy is proper).
\end{itemize}

We call a communication protocol \emph{efficient}
if its cost is poly-logarithmic in the length of input.

\subsection{Read-$k$ families of functions}

Let us consider the following model of dependence among random variables.

\begin{definition}[Read-$k$ families]
  Let~$X_1,\dots,X_m$ be independent random variables.
  For~$j \in [r]$, let~$P_j\subseteq[m]$ and let~$f_j$ be a Boolean function of~$(X_i)_{i\in P_j}$.
  If every~$i\in [m]$ belongs to at most~$k$ among the~$r$ sets~$P_1,\dots,P_r$,
  then the random variables~$Y_j=f_j((X_i)_{i \in P_j})$ are called a \emph{read-$k$ family}.
\end{definition}

The following lemma is due to Finner~\cite{F92_A_Ge}.

\begin{lemma}[Finner~\cite{F92_A_Ge}]
  \label{l_Hold}
  Let~$Y_1,\dots,Y_n$ be a read-$k$ family of random variables taking non-negative values.
  Then
  \[
    \E*{\prod_{i=1}^nY_i}\le\prod_{i=1}^n\sqrt[k]{\E{Y_i^k}}.
  \]
\end{lemma}

Note that the generalized H\"older inequality implies that~%
$\E{\prod Y_i}\le\prod\sqrt[n]{\E{Y_i^n}}$
in general, without making any independence assumption.
This corresponds to choosing~$k=n$ in the lemma.
On the other hand, when~$k=1$ (i.e., $Y_1,\dots,Y_n$ are mutually independent),
the expectation of their product equals the product of their expectations.
Accordingly, Lemma~\ref{l_Hold} gives a natural interpolation
between these two extreme cases.

\section{Hierarchy of shared randomness in multi-party protocols}
  \label{s_srand}

When there are more than two players,
it is possible to give the players access to shared randomness in several different ways.

Most naturally, the parties may have free access to the same string of random bits---%
we call this mode \emph{unrestricted shared randomness}.
Note that this mode of shared randomness is often implicitly assumed to be available to the players;
for example, unrestricted shared randomness is required in order to be able to use \emph{mixed strategies},
and therefore applicability of the minimax principle in multi-party communication depends on it.

Let~$k\ge2$ be the number of players.
For every~$t\in\set{2,\dots,k}$, we can define the mode of \emph{$t$-shared randomness},
where every~$t$ players share their own string of random bits.
The $k$-shared randomness is the same thing as the unrestricted shared randomness.
Sometimes we will refer to the~$(k-1)$-shared mode as \emph{randomness on the forehead}.

We will also consider \emph{XOR-shared randomness},
where every player~$\A_i$ is given access to an arbitrarily long random string~$r_i$,
such that every~$(k-1)$ strings~$r_i$ are uniform and mutually independent
but the bitwise XOR of~$r_1,\dots,r_k$ is~$0$ everywhere.

If we consider the case of~$k=2$, we can see that XOR-shared and $2$-shared modes are the same.
For~$k=3$, we already have three modes of shared randomness:
XOR-shared, $2$-shared and $3$-shared.
We will see below that these three modes offer different computational power.

In general, $t$-shared randomness is always at least as strong as~$(t-1)$-shared randomness,
as the latter can always be emulated using the former.
Also, XOR-shared randomness can be emulated in the $2$-shared mode;
to do that, let $r_i$ equal the bit-wise XOR
of the random string shared between~$\A_{i-1}$ and~$\A_i$
and the random string shared between~$\A_i$ and~$\A_{i+1}$,
where the subscripts are interpreted modulo~$k$.
Now the strings $r_1,\dots,r_k$ are distributed
as required by the definition of XOR-shared mode.

One interesting example
that demonstrates usefulness of XOR-shared randomness when~$k\ge3$
is the \emph{multi-party equality function};
that is, the total Boolean function of~$k$ arguments~$x_1,\dots,x_k$
that takes value~$1$ if and only if~$x_1=x_2=\dots=x_k$.

\begin{claim}
  \label{c_k-eq}
  For any~$c\in\NN$, there exists a classical protocol for the~$k$-party equality function,
  where XOR-shared randomness is used, each player sends~$c$ bits to the referee,
  and the following holds:
  \begin{itemize}
  \item If~$x_1=x_2=\dots=x_k$, then the referee's answer is always~$1$;
  \item otherwise, the referee's answer is~$0$ with probability~$1-1/2^c$.
  \end{itemize}
\end{claim}

\begin{proof}
  For~$r,x\in\zo^n$, let~$r\cdot x$ be the inner product of~$r$ and~$x$ in finite field~$\GF(2)$:
  $r\cdot x\deq\bigoplus_{i:r(i)=1}x(i)$.
  Consider the following protocol~$\Pk$, where the players use XOR-shared randomness
  and each of them sends a single bit to the referee:
  \begin{enumerate}[1.]
  \item
    For all~$i\in[k]$, the~$i$th player uses his random string~$r_i\in\zo^n$
    and computes~$m_i\deq r_i\cdot x_i$, then sends~$m_i$ to the referee.
  \item
    The referee outputs~$\neg(m_1\oplus\dots\oplus m_k)$.
  \end{enumerate}

  By the definition of XOR-shared randomness, we have that~$r_1\oplus\dots\oplus r_k=\bar0$.
  Therefore~$(r_1\oplus\dots\oplus r_k)\cdot x_k=0$, and we can write
  \begin{align}
    &
    m_1\oplus\dots\oplus m_k
    \nonumber \\
    &=r_1\cdot x_1\oplus\dots\oplus r_k\cdot x_k
    \nonumber \\
    &=r_1\cdot x_1\oplus\dots\oplus r_k\cdot x_k
    \oplus(r_1\oplus\dots\oplus r_k)\cdot x_k
    \nonumber \\
    &=r_1\cdot(x_1\oplus x_k)\oplus\dots\oplus r_{k-1}\cdot(x_{k-1}\oplus x_k).
    \label{m_keq}
  \end{align}

  Note that~$r_1,\dots,r_{k-1}$ is a uniformly random~$(k-1)$-tuple of~$n$-bit strings,
  and therefore the rightmost part of~(\ref{m_keq}) equals~$1$ with probability exactly~$1/2$
  if at least one of the~$k-1$ values~$x_1\oplus x_k,\dots,x_{k-1}\oplus x_k$ is different from~$\bar0$.
  If, on the other hand, $x_1=x_2=\dots=x_k$ then(\ref{m_keq}) equals~$0$ with certainty.

  Accordingly, $\Pk$ outputs ``$1$'' whenever $x_1=x_2=\dots=x_k$,
  and otherwise it outputs ``$0$'' with probability~$1/2$.
  To get a protocol as promised by our claim, we can run~$c$ independent instances of~$\Pk$ in parallel
  and output ``$1$'' if and only if all~$c$ instances answered ``$1$''.
\end{proof}

We are now prepared to prove that the modes of shared randomness form a proper hierarchy when~$k\ge3$.

\begin{reptheorem}{t_hie}
  \theohie
\end{reptheorem}

\begin{proof}
  To prove the first part of the theorem,
  consider the~$t$-party problem~$\GP_t$,
  letting the players~$\A_1,\dots,\A_t$ receive the corresponding fragments of input.
  (The other~$(k-t)$ players do not receive any input.)
  Theorem~\ref{t_main}, which will be proved in the next section,
  implies the result in this case.

  The second part follows from considering the two-party equality problem,
  when the input is distributed between~$\A_1$ and~$\A_2$
  and the other players do not receive any input.
  It is clear that this problem can be solved with constant cost
  in the classical SMP model with $2$-shared randomness.
  Now suppose that there exists a protocol of cost~$c$
  in the classical SMP model with XOR-shared randomness,
  and we will prove that~$c=\Omega(\sqrt{n})$.
  Note that in this model, the random strings given to~$\A_1$ and~$\A_2$
  are uniform and independent,
  although they are correlated with the random strings given to the other players.
  Such a protocol can be transformed without changing the cost
  to a protocol in the two-party classical SMP model
  where the two players do not share randomness but each player shares randomness with the referee,
  because in the latter model, the referee can generate all the messages
  which would have been generated by players~$\A_3,\dots,\A_k$.
  By the same technique used by Newman~\cite{N91_Pri},
  this protocol can be further transformed to a protocol of cost~$O(c+\log n)$
  in the two-party classical SMP model without shared randomness at all.
  It was shown by Newman and Szegedy~\cite{NS96_Pub}
  that the communication complexity of the equality problem in this model
  is~$\Omega(\sqrt n)$, and therefore~$c$ must be~$\Omega(\sqrt{n})$.

  The third part follows from considering the~$k$-party equality function.
  The upper bound is shown in Claim~\ref{c_k-eq}.
  The lower bound follows from Newman and Szegedy~\cite{NS96_Pub},
  because any~$k$-party SMP protocol without shared randomness among~$\A_1,\dots,\A_k$
  for the~$k$-party equality function
  can be used to construct a two-party SMP protocol without shared randomness
  between two players~$\A'_1$ and~$\A'_2$
  for the two-party equality function without affecting its cost:
  $\A'_1$ simulates~$\A_1$, and~$\A'_2$ simulates~$\A_2,\dots,\A_k$.
\end{proof}

\begin{remark}
  Besides its own elegance, the hierarchy of shared randomness is a useful technical tool.
  For our lower bound proof for the quantum complexity of~$\GPk$
  (which is the main technical result of Section~\ref{s_vs}),
  we need different modes of shared randomness.
  Informally speaking, we use a hybrid argument that puts certain restrictions on the input values,
  and those restrictions inevitably create shared randomness of certain type that becomes available to the players.
  We show that the sort of randomness that is introduced
  corresponds to one of the restricted modes of shared randomness,
  whose availability does not make the communication problem easy for quantum communication.
\end{remark}

\section{Shared randomness vs.\ quantum communication}
  \label{s_vs}

In this section, we will analyze the complexity of~$\GPk$
to compare the resource of shared randomness
to that of quantum communication in multi-party protocols.

Fix~$k\ge3$.
Recall that in the classical SMP model with shared randomness~$\GPk$ has a trivial solution,
where each player sends one bit to the referee.

As a warm-up, consider the case of quantum protocols \emph{without shared randomness}.%
\footnote{%
  We shall see soon why in the actual proof
  we have to consider different modes of shared randomness,
  even if our only purpose was to get a lower bound on the complexity of~$\GPk$
  in the quantum model without shared randomness.}
Let~$\Pk$ be a quantum protocol that communicates~$c$ qubits and solves~$\GPk$,
and let~$\Uk$ be the uniform distribution over~$k$-tuples~$(x_1,\dots,x_k)\in\zo^{n\times k}$.
Now consider the behavior of~$\Pk$ when the input is distributed according to~$\Uk$
(note that such input is almost never valid for~$\GPk$).

For~$i\in[k]$, let~$\sigma_i$ be a density matrix representing the (mixed) state
that the referee receives from~$\A_i$ when the input distribution is~$\Uk$.
Since the senders share no randomness and~$\Uk$ is a product distribution,
the state of the referee before his measurement is performed can be written as~%
$\sigma_1\otimes\dots\otimes\sigma_k$.

For~$i\in[k]$, let~$X_i$ be an~$n$-bit random string taking the value of input~$x_i$.
By Claim~\ref{c_rac}, there exists~$j_0 \in [n]$ such that~%
$\sum_{i=1}^k\norm{\sigma_0^i-\sigma_1^i}_1^2
 \in O(c/n)$,
where~$\sigma_a^i$ is the message from~$\A_i$, conditional upon~$X_i(j_0)=a$.

Since the random variables~$X_1(j_0),\dots,X_k(j_0)$ are mutually independent
and each~$X_i(j_0)$ can be correlated only with~$\sigma_i$,
Claim~\ref{c_xor} implies that the referee can predict~$X_1(j_0)\oplus\dots\oplus X_k(j_0)$
with probability at most~%
$1/2+(1/2^{k+1})\prod_{i=1}^k\norm{\sigma_0^i-\sigma_1^i}_1$,
which is at most~$1/2+O(c/(kn))^{k/2}$
by the inequality of arithmetic and geometric means.

Note that this guarantees that the ``advantage over random guess'' that the referee can have
in predicting~$X_1(j_0)\oplus\dots\oplus X_k(j_0)$
using the messages received from the players is~$o(1/n)$,
as long as~$k\ge3$ and~$c\in o(kn^{1-2/k})$.
Moreover, similar reasoning can be applied to conclude
that for most of the values of~$j_0\in[n]$,
the possible advantage in predicting~$X_1(j_0)\oplus\dots\oplus X_k(j_0)$ must be very small.

With this observation in hand, we would like to apply a ``hybrid-like'' reasoning,
arguing that in order to distinguish
between those inputs where most of bitwise XORs equal~$0$ and those where most equal~$1$,
a protocol should be able, informally, to ``accumulate advantage'' from different input positions.
We would like to claim that this is impossible as long as the advantage is negligible for almost every~$j_0\in[n]$.

Here comes the main subtlety of our proof.
Note that using hybrid-like argument puts a condition on a part of the input:
specifically, in order for the ``hybrid scenario'' to get through,
it has to be argued that it is hard for the protocol to predict most of the values of~$X_1(j)\oplus\dots\oplus X_k(j)$,
even if the players ``know'' the values of~$X_1(j')\oplus\dots\oplus X_k(j')$
for those positions~$j'$ that were considered in the earlier stages of the induction.
But such conditioning creates certain type of shared randomness between the players,
and we can no longer assume mutual independence of the messages received by the referee,
as we have done in the reasoning above.

Recall that we are dealing with a communication problem
that is easy in the presence of shared randomness even classically.
How can we hope for quantum hardness,
as required for the hybrid argument to be applicable?
It turns out that \emph{the mode of shared randomness that results from using the hybrid method
is not powerful enough to make the problem easy, even for quantum communication}.
Our proof of Lemma~\ref{l_hyb-step} below follows rather closely the outline given above,
but it also contains some new ingredients
required to make the argument robust against weaker modes of shared randomness.

\subsection{Exponential Separation for multi-party protocols}
  \label{ss_proof}

We are ready to prove our main technical statement.

\begin{reptheorem}{t_main}
  \theomain
\end{reptheorem}

First, we set up notation to describe protocol~$\Pk$ which uses~$t$-shared randomness.
Let~$\Vkt=\{S \subseteq [k]\colon |S|=t\}$.
For any~$S \in \Vkt$, let~$R_S$ be the random string (of arbitrary length)
shared by the players~$\A_i$ for~$i \in S$.
Then~$\A_i$ holds the~$R_S$'s for all~$S \in \Vkt$ containing~$i$.
For convenience, let~$\tilde{R}_i=(R_S)_{i \in S \in \Vkt}$
be the~$\binom{k-1}{t-1}$-tuple of these random strings.
For each instantiation~$R_S=r_S$,
let~$\tilde{r}_i=(r_S)_{i \in S \in \Vkt}$ be the corresponding instantiation of~$\tilde{R}_i$.
In addition, let~$\vR=(R_S)_{S \in \Vkt}$ be the~$\binom{k}{t}$-tuple of all shared random strings,
and let~$\vr=(r_S)_{S \in \Vkt}$ be any instantiation of~$\vR$.
For each~$i \in [k]$, player~$\A_i$ sends a quantum state~$\rho^i_{x_i,\tilde{r}_i}$
conditional upon receiving input~$x_i$ and random strings~$\tlR_i=\tlr_i$.
Let~$c_i$ be the length of the quantum message sent by~$\A_i$;
i.e., the length of state~$\rho^i_{x_i,\tlr_i}$ is~$c_i$ qubits.
By assumption, $c \deq \sum_{i=1}^k c_i = o(kn^{1-t/k})$.

As before, let~$\Uk$ be the uniform distribution over~$k$-tuples~$(x_1,\dots,x_k)\in\zo^{n\times k}$.
To prove the theorem, we will use the following lemma.

\begin{lemma}
  \label{l_hyb-step}
  Let~$2\le t\le k-1$,
  and let~$\Pk$ be a quantum SMP protocol of cost~$c=o(kn^{1-t/k})$
  that uses~$t$-shared randomness as defined above.
  Suppose the player~$\A_i$ receives the random input~$X_i$ for~$(X_1,\dots,X_k) \in \U(\zonk)$.
  Then there exists~$J \subseteq [n]$ of size at least~$2n/3$
  such that for every~$j \in J$ a referee who is allowed to apply an arbitrary quantum measurement
  to the messages received according to~$\Pk$
  can predict the value of~$X_1(j)\oplus\dots\oplus X_k(j)$
  with probability at most~$1/2+o(1/n)$.
\end{lemma}

\begin{proof}
  Let~$\sigma^i_{a,\vr}(j)$ be the expectation of~$\rho^i_{X_i,\tlR_i}$
  conditional upon~$X_i(j)=a$ and~$\vR=\vr$,
  for any~$i \in [k]$, $j \in [n]$, $a \in \zo$ and possible~$\vr$.
  Define~$\alpha_{i,\vr} \in \mathbb{R}^n$ as
  \begin{equation}
    \alpha_{i,\vr}(j)=\frac12 \norm{\sigma^i_{0,\vr}(j)-\sigma^i_{1,\vr}(j)}_1.
    \label{eq:defaivj}
  \end{equation}
  Then by Claim~\ref{c_rac},
  \[
    \sum_{j=1}^n (\alpha_{i,\vr}(j))^2 \le O(c_i).
  \]
  Taking the sum of both sides over~$i \in [k]$ and using~$\sum_{i=1}^k c_i=c$ yields
  \[
    \sum_{i=1}^k\sum_{j=1}^n (\alpha_{i,\vr}(j))^2 \le O(c).
  \]
  This holds for any possible~$\vr$. So
  \[
    \E*{\sum_{i=1}^k\sum_{j=1}^n (\alpha_{i,\vR}(j))^2}\le O(c).
  \]
  (Here the expectation is taken with respect to the~$R_S$'s).
  Thus, there exists some~$J \subseteq [n]$ of size at least~$2n/3$
  such that, for any~$j_0 \in J$,
  \begin{equation}
    \E*{\sum_{i=1}^k (\alpha_{i,\vR}(j_0))^2} \le O\left(\frac{c}{n}\right).
    \label{eq:sum_aivij}
  \end{equation}
  Let~$\sigma_{a,\vr}(j_0)$ be the expectation
  of~$\rho^1_{X_1,\tlR_1} \otimes \dots \otimes \rho^k_{X_k,\tlR_k}$
  conditional upon~$X_1(j_0)\oplus \dots \oplus X_k(j_0)=a$ and~$\vR=\vr$,
  for any~$a \in \zo$ and possible~$\vr$.
  Since the~$X_i(j_0)$'s are i.i.d.\ with~$X_i(j_0)\sim \U(\zo)$,
  and they are also independent from the~$R_S$'s, we have
  \[
    \sigma_{a,\vr}(j_0)=\E_{a_1 \oplus \dots \oplus a_k=a}
    {\sigma^1_{a_1,\tlr_1}(j_0) \otimes \dots \otimes \sigma^k_{a_k,\tlr_k}(j_0)}.
  \]
  (Here the expectation is taken with respect to the~$a_i$'s).
  So by Claim~\ref{c_xor} and~(\ref{eq:defaivj}) we get
  \[
    \norm{\sigma_{0,\vr}(j_0)-\sigma_{1,\vr}(j_0)}_1
    = 2 \prod_{i=1}^k \alpha_{i,\vr}(j_0).
  \]
  Now let~$\sigma_{a}(j_0)$ be the expectation
  of~$\rho^1_{X_1,\tlR_1} \otimes \dots \otimes \rho^k_{X_k,\tlR_k}$
  conditional upon~%
  $X_1(j_0)\oplus \dots \oplus X_k(j_0)=a$, for~$a \in \zo$.
  Then~$\sigma_{a}(j_0)=\E{\sigma_{a,\vR}(j_0)}$.
  Thus,
  \begin{align*}
    \norm{\sigma_{0}(j_0)-\sigma_{1}(j_0)}_1
    &=\norm{\E{\sigma_{0,\vR}(j_0)}-\E{\sigma_{1,\vR}(j_0)}}_1 \\
    &\le
    \E{\norm{\sigma_{0,\vR}(j_0)-\sigma_{1,\vR}(j_0)}_1} \\
    &= 2\E*{\prod_{i=1}^k \alpha_{i,\vR}(j_0)}.
  \end{align*}
  Note that~$Z_i\deq\alpha_{i,\tlR_i}(j_0)$ is a non-negative function of the $R_S$'s
  for~$i \in S \in \Vkt$.
  (Recall that~$\tlR_i=(R_S)_{i \in S \in \Vkt}$.)
  Since the~$R_S$'s are independent random variables,
  and each~$R_S$ is read~$t$ times (by the~$Z_i$'s for $i \in S$),
  we know that~$Z_1,\dots,Z_k$ form a read-$t$ family.
  Thus, by invoking Lemma~\ref{l_Hold}, we get
  \[
    \E* {\prod_{i=1}^k \alpha_{i,\vR}(j_0)}
    \le \left(\prod_{i=1}^k \E{(\alpha_{i,\vR}(j_0))^t}\right)^{1/t}.
  \]
  Since~$0 \le \alpha_{i,\vR}(j_0) \le 1$ and~$t\ge2$, we have
  \[
    \left(\prod_{i=1}^k \E{(\alpha_{i,\vR}(j_0))^t}\right)^{1/t}
    \le \left(\prod_{i=1}^k \E{(\alpha_{i,\vR}(j_0))^2}\right)^{1/t}.
  \]
  Then by the inequality of arithmetic and geometric means and~(\ref{eq:sum_aivij}),
  \begin{align*}
    {\left(\prod_{i=1}^k \E{(\alpha_{i,\vR}(j_0))^2}\right)\!\!}^{1/t}
    &\le
    {\left(\frac{1}{k}\sum_{i=1}^k \E{(\alpha_{i,\vR}(j_0))^2}\right)\!}^{k/t}
    \\
    &\le
    \left( O\left( \frac{c}{kn}\right)\right)^{k/t}
    = o\left(\frac{1}{n}\right),
  \end{align*}
  provided~$c=o(kn^{1-t/k})$.
  So, we have~$\norm{\sigma_{0}(j_0)-\sigma_{1}(j_0)}_1=o(1/n)$.
  Namely, the referee can predict the value of~$X_1(j_0) \oplus \dots \oplus X_k(j_0)$
  with probability at most~$1/2+o(1/n)$.
  This holds for any $j_0 \in J$.
\end{proof}

\begin{proof}[Proof of Theorem~\ref{t_main}]
  Let~$\Pk$ be a quantum SMP protocol of cost~$c=o(kn^{1-t/k})$
  that uses~$t$-shared randomness.
  We will show that there exist~$L=\floor{3n/4}$ coordinates%
  \footnote{%
    In fact, our statement holds for any~$L \le (1-\varepsilon)n$,
    where~$\varepsilon$ can be any small constant.}
  $j_1,\dots,j_L \in [n]$ satisfying the following conditions.
  For~$l \in [L]$ and~$a \in \zo$,
  let~$W_{l,a}$ be the set of~$k$-tuples~$(x_1,\dots,x_k) \in \zonk$
  satisfying~$x_1(j) \oplus \dots \oplus x_k(j)=a$ for~$j=j_1,j_2,\dots,j_l$,
  and let~$\tau_{l,a}=\E{\rho^1_{X_1,\tlR_1}\otimes\dots \otimes \rho^k_{X_k,\tlR_k}}$
  for~$(X_1,\dots,X_k) \sim \U(W_{l,a})$
  (here the expectation is taken with respect to the~$X_i$'s and~$R_S$'s).
  Then:
  (i) $\norm{\tau_{1,0}-\tau_{1,1}}_1=o(1/n)$;
  (ii) for any~$l \in [L-1]$, $\norm{\tau_{l,0}-\tau_{l+1,0}}_1=o(1/n)$
  and~$\norm{\tau_{l,1}-\tau_{l+1,1}}_1=o(1/n)$.
  If this is true, then by the triangle inequality,
  \begin{align}
    &\norm{\tau_{L,0}-\tau_{L,1}}_1
    \nonumber \\
    & \le
    \sum_{l=1}^{L-1}\norm{\tau_{l,0}-\tau_{l+1,0}}_1
    +\sum_{l=1}^{L-1}\norm{\tau_{l,1}-\tau_{l+1,1}}_1
    \nonumber \\
    &\quad
    +\norm{\tau_{1,0}-\tau_{1,1}}_1
    \nonumber \\
    &=o(L/n)
    \nonumber \\
    &=o(1).
    \label{eq:triseq}
  \end{align}
  On the other hand, for any~$(x_1,\dots,x_k) \in W_{L,0}$,
  it holds that~$|x_1 \oplus x_2 \oplus \dots \oplus x_k|\le n-L < n/3$
  and hence~$\GPk(x_1,\dots,x_k)=0$.
  Similarly, for any~$(x_1,\dots,x_k) \in W_{L,1}$,
  it holds that~$|x_1 \oplus x_2 \oplus \dots \oplus x_k|\ge L > 2n/3$
  and hence~$\GPk(x_1,\dots,x_k)=1$.
  Therefore, if the referee can correctly predict the value of~$\GPk(x_1,\dots,x_k)$
  on any~$(x_1,\dots,x_k) \in W_{L,0} \sqcup W_{L,1}$,
  then he should be able to distinguish between~$\tau_{L,0}$ and~$\tau_{L,1}$
  with probability at least $2/3$, which implies that~$\norm{\tau_{L,0}-\tau_{L,1}}_1=\Omega(1)$.
  But this is contradictory to~(\ref{eq:triseq}).
  So the referee must fail to solve~$\GPk$
  on some valid input from~$W_{L,0}$ or~$W_{L,1}$.

  To find the desired~$j_1,\dots,j_L$,
  we use one initial step and~$L-1$ inductive steps as follows.

  \textbf{Initial Step:}
  Consider the behavior of~$\Pk$ on the random input~$(X_1,\dots,X_k) \in \U(\zonk)$.
  By a straightforward application of Lemma~\ref{l_hyb-step},
  there exists~$J \subseteq [n]$ of size at least~$2n/3$
  such that, for any~$j \in J$,
  the referee can predict the value of~$X_1(j) \oplus \dots \oplus X_k(j)$
  with probability at most~$1/2+o(1/n)$.
  In other words, for any~$j \in J$, we have~%
  $\norm{\sigma_0(j)-\sigma_1(j)}_1=o(1/n)$
  where~$\sigma_{a}(j)$ is the expectation
  of~$\rho^1_{X_1,\tlR_1} \otimes \dots \otimes \rho^k_{X_k,\tlR_k}$
  conditional upon~$X_1(j)\oplus \dots \oplus X_k(j)=a$, for~$a \in \zo$.
  Set~$j_1$ to be any~$j \in J$.
  Then~$\norm{\tau_{1,0}-\tau_{1,1}}_1=\norm{\sigma_0(j)-\sigma_1(j)}_1=o(1/n)$ as desired.

  \textbf{Inductive Step:}
  Suppose now we have fixed~$j_1,\dots,j_l \in [n]$ for some~$l \le 3n/4$.
  Let~$T=\{j_1,\dots,j_l\}$ and~$T^c=[n] \setminus T$.

  Let us consider the behavior of~$\Pk$ on the random input~$(X_1,\dots,X_k) \sim \U(W_{l,0})$
  (recall that~$W_{l,0}$ is the set of~$(x_1,\dots,x_k) \in \zonk$
  satisfying~$(x_1 \oplus \dots \oplus x_k)_T=\bar{0}$).
  Note that the~$X_i$'s are not completely independent but only~$(k-1)$-wise independent.
  So there exists some correlation among the inputs to different players,
  which might be exploited to gain some advantage.
  However, we will show that, even in this case,
  the referee can still predict the value of~$X_1(j) \oplus \dots \oplus X_k(j)$
  with probability at most~$1/2+o(1/n)$ for at least~$2/3$ fraction of~$j \in T^c$.

  Let~$Y_i=(X_i)_T$.
  Then~$(Y_1,\dots,Y_k)$ is uniformly distributed
  among all~$(y_1,\dots,y_k)\in \zolk$ satisfying~$y_1 \oplus \dots \oplus y_k=\bar{0}$.
  Namely, $(Y_1,\dots,Y_k)$ can be viewed as some XOR randomness
  (which is a special kind of~$2$-shared randomness) shared by the players.
  Also, note that the~$(X_i)_{T^c}$'s are i.i.d.\ with~$(X_i)_{T^c} \sim \U(\zo^{n-l})$,
  and they are also independent from the~$Y_i$'s.
  Finally, the~$Y_i$'s and~$(X_i)_{T^c}$'s are all independent from the~$R_S$'s.

  Now consider the following protocol~$\Pk'$ which attempts to solve~$\GPk$ for~$(n-l)$-bit strings.
  The players share XOR randomness~$(Y_1',\dots,Y_k')$
  which has the same distribution as~$(Y_1,\dots,Y_k)$.
  In addition, they also share~$t$-shared randomness~$(R_S')_{S \in \Vkt}$
  which has the same distribution as~$(R_S)_{S \in \Vkt}$.
  Furthermore, the~$Y_i'$'s and~$R_S'$'s are independent.
  Now suppose player~$\A_i$ receives input~$x_i' \in \zo^{n-l}$ and random strings~$Y_i'=y_i$ and~$\tlR_i'=\tlr_i$.
  Then~$\A_i$ first finds the unique~$x_i \in \zon$ such that~$(x_i)_T=y_i$ and~$(x_i)_{T^c}=x_i'$,
  and then sends the quantum message~$\rho^i_{x_i,\tlr_i}$ of~$\Pk$ to the referee.

  Since the~$(Y_1',\dots,Y_k')$ is a special kind of~$2$-shared randomness,
  $\Pk'$ uses only~$t$-shared randomness.
  So by Lemma~\ref{l_hyb-step}
  (replacing the original~$n$ by~$n-l$ and noting that~$n-l=\Theta(n)$, since~$l \le 3n/4$),
  we know that, on the random input~$(X_1',\dots,X_k') \sim \U(\zo^{(n-l)\times k})$,
  there exists~$J' \subseteq [n-l]$ of size at least~$2(n-l)/3$
  such that for any~$j \in J'$ the referee can predict the value
  of~$X_1'(j)\oplus \dots \oplus X_k'(j)$
  with probability at most~$1/2+o(1/n)$
  using the messages received according to~$\Pk'$.
  Meanwhile, by the construction of~$X_i'$'s, $Y_i'$'s, $R_S'$'s and~$\Pk'$,
  it is obvious that the joint message sent according to~$\Pk'$
  has the same distribution as~%
  $\rho^1_{X_1,\tlR_1}\otimes\dots \otimes\rho^k_{X_k,\tlR_k}$.
  In addition, the bits of~$X_i'$ are in one-to-one correspondence with the bits of~$(X_i)_{T^c}$.
  Thus, getting back to the original protocol~$\Pk$,
  we know that, on the random input~$(X_1,\dots,X_k) \sim \U(W_{l,0})$,
  there exists~$J_0 \subseteq T^c$ of size at least~$2(n-l)/3$ (corresponding to~$J'$)
  such that for any~$j \in J_0$
  the referee can predict the value of~$X_1(j)\oplus \dots \oplus X_k(j)$
  with probability at most~$1/2+o(1/n)$
  using the messages received according to~$\Pk$.
  So for any~$j \in J_0$, we have~$\norm{\sigma_0(j)-\sigma_1(j)}_1=o(1/n)$,
  where $\sigma_a(j)$ is the expectation of~%
  $\rho^1_{X_1,\tlR_1}\otimes\dots \otimes\rho^k_{X_k,\tlR_k}$
  conditional upon~$X_1(j)\oplus \dots \oplus X_k(j)=a$ for~$a \in \zo$.
  Now if we set~$j_{l+1}=j$,
  then depending on the value of~$x_1(j)\oplus \dots \oplus x_k(j)$,
  $W_{l,0}$ is split into to two equal-sized subsets:
  $W_{l+1,0}$ and~$W_{l,0}\setminus W_{l+1,0}$.
  It follows that~$\tau_{l,0}=(\sigma_0(j)+\sigma_1(j))/2$ and~$\tau_{l+1,0}=\sigma_0(j)$, and hence~%
  $\norm{\tau_{l,0}-\tau_{l+1,0}}_1
  =\norm{\sigma_{0}(j)-\sigma_{1}(j)}_1/2=o(1/n)$.

  By a similar argument, we can also prove that,
  on the random input~$(X_1,\dots,X_k) \sim \U(W_{l,1})$,
  there also exists~$J_1 \subseteq T^c$ of size at least~$2(n-l)/3$
  such that for any~$j \in J_1$ the referee can predict~$X_1(j) \oplus \dots \oplus X_k(j)$
  with probability at most~$1/2+o(1/n)$.
  Then, if we set~$j_{l+1}$ to be any~$j \in J_1$,
  then we get~$\norm{\tau_{l,1}-\tau_{l+1,1}}_1=o(1/n)$.

  Now since~$|J_0|,|J_1| \ge 2(n-l)/3$, $J_0 \cap J_1$ must be non-empty.
  We set~$j_{l+1}$ to be any~$j \in T_0 \cap T_1$.
  Then we achieve both~%
  $\norm{\tau_{l,0}-\tau_{l+1,0}}_1=o(1/n)$
  and~%
  $\norm{\tau_{l,1}-\tau_{l+1,1}}_1=o(1/n)$.

  Iterate this inductive step~$L-1$ times,
  and in the end we obtain the desired~$j_1,\dots,j_L$.
  This completes a proof of the theorem.
\end{proof}

\section*{Acknowledgments}

The authors thank anonymous reviewers for helpful comments on previous versions of this paper.
The authors acknowledge support by ARO/NSA under grant W911NF-09-1-0569.

\bibliography{tex}

\end{document}